\documentclass[journal]{IEEEtran} 

\usepackage[utf8]{inputenc}
\usepackage{amsmath,amssymb,amsthm, bbm}
\usepackage{algorithm,algpseudocode}
\usepackage{xcolor}
\usepackage{array}
\usepackage{float, placeins}
\usepackage{enumitem}
\usepackage{multirow}
\usepackage{cite}
\usepackage[normalem]{ulem}
\usepackage{mathrsfs}
\usepackage{graphicx}
\usepackage{caption}
\usepackage{tabularx}
\usepackage{setspace}
\usepackage{comment}
\usepackage{color}
\usepackage{hyperref}
\allowdisplaybreaks
\theoremstyle{remark}
\newtheorem{assum}{Assumption}

\theoremstyle{definition}
\newtheorem{definition}{Definition}[section]

\newtheorem{remark}{Remark}
\newtheorem{theorem}{Theorem}
\newtheorem*{theorem*}{Theorem}

\newtheorem*{result*}{Result}

\newcommand{\prox}{\mathop{\rm prox}\nolimits}
\newcommand{\Real}{\mathbb{R}}
\newcommand{\Tra}{^{\sf T}} % Transpose

\newcommand{\amp}{\mathop{\:\:\,}\nolimits}
\usepackage{xspace}
\definecolor{blue}{rgb}{0.2,0.5,0.7}
\definecolor{green}{rgb}{0.3,0.68,0.29}
\definecolor{purple}{rgb}{0.6,0.31,0.64}

\title{Proximal Hamiltonian Monte Carlo}
\author{Apratim~Shukla, 
        Dootika~Vats, 
        and~Eric~C.~Chi%
        \thanks{Apratim Shukla and Dootika Vats are with the Department of Mathematics and Statistics, Indian Institute of Technology Kanpur, India (e-mail:~apratims21@iitk.ac.in; dootika@iitk.ac.in).}%
        \thanks{Eric C. Chi is with the School of Statistics, University of Minnesota, Minneapolis, MN, USA (e-mail:~echi@umn.edu).}%
}
\date{\today}

\begin{document}
\singlespacing
\maketitle

%\tableofcontents

\begin{abstract}
Bayesian formulation of modern day signal processing problems has called for improved Markov chain Monte Carlo (MCMC)  sampling algorithms for inference. The need for efficient sampling techniques has become indispensable for high dimensional distributions that often characterize many core signal processing problems, e.g., image denoising, sparse signal recovery, etc. A major issue in building effective sampling strategies, however, is the non-differentiability of the underlying posterior density. 
% This is a direct consequence of employing a non-differentiable prior. 
Such posteriors are popular in models designed to recover sparse signals. As a result, the use of efficient gradient-based MCMC sampling techniques becomes difficult. We circumvent this problem by proposing a Proximal Hamiltonian Monte Carlo (p-HMC) algorithm, which leverages elements from convex optimization like proximal mappings and Moreau-Yosida (MY) envelopes within Hamiltonian dynamics. Our method improves upon the current state of the art non-smooth Hamiltonian Monte Carlo as it achieves a relatively sharper approximation of the gradient of log posterior density and a computational burden of at most the current state-of-the-art. A chief contribution of this work is the theoretical analysis of p-HMC. We provide conditions for geometric ergodicity of the underlying HMC chain. On the practical front, we propose guidance on choosing the key p-HMC hyperparameter -- the regularization parameter in the MY-envelope. We demonstrate p-HMC's efficiency over other MCMC algorithms on benchmark problems of logistic regression and low-rank matrix estimation.
\end{abstract}


\section{Introduction}
\IEEEPARstart{M}{odern} computational power has paved the way for increasingly complex models in statistical signal processing \cite{manolakis2005statistical}. Examples include models for compressed sensing \cite{MR2241189}, signal denoising \cite{MR1379464, MR1311089, costa2015sparse}, image denoising \cite{doi:10.1137/120874989}, image deconvolution \cite{chaux2006iterative}, superresolution \cite{park2003super,doi:10.1137/15M1016552}, magnetic resonance image reconstruction \cite{lustig2007sparse}, and computer vision \cite{10.1007/978-3-540-24673-2_23}, to name a few. A defining feature of these models is that they admit algorithms that recover low-dimensional characteristics from high-dimensional data. Within these model spaces, Bayesian models that employ sparsity inducing penalties furnish both point estimates of model parameters as well as desirable uncertainty estimates.
% about them. Such Bayesian models, however, often possess intractable posterior distributions, whose characteristics need to be computed to quantify model parameter uncertainty.
Inference for Bayesian signal processing models depends on reliably  generating samples from an often intractable posterior distribution. While generating  independent and identically distributed (iid) samples from such complex and high-dimensional posterior distributions would be ideal, it is usually not possible beyond small dimensions. Fortunately, we can generate samples using a Markov chain  whose stationary and limiting distribution coincides with our desired posterior distribution via Markov chain Monte Carlo (MCMC) methods \cite{1188770, MR2080278}. These methods have found wide-spread use in signal-processing including radar processing \cite{7089944}, sequential data processing \cite{MR2766856}, audio restoration \cite{o1994interpolation}, image sequence restoration \cite{559620} and various other signal/image processing problems \cite{FITZGERALD20013, 1550195, 6952466, 6945861, Luengosurvey}. The quality of Bayesian inference across applications relies critically on the quality of the MCMC sampler employed. 
%The simplest method is perhaps the random walk Metropolis (RWM) \cite{metr:1953} which is popular owing to its easy implementability.  For specific problems involving Gaussian priors, elliptical slice sampling \cite{murray2010elliptical} and high-dimensional Gaussian sampling by \cite{6158583} have been proposed. In general, to sample from high-dimensions, gradient-based MCMC schemes have proved to be an effective alternative. These constitute Metropolis-Adjusted Langevin algorithm (MALA) \cite{roberts1996exponential}, the Barker proposal \cite{livingstone2022barker} and Hamiltonian Monte Carlo (HMC) \cite{hanson2001markov, neal2011mcmc}. These schemes exploit the information from the target distribution via gradients to efficiently explore the state-space.

Metropolis-Hastings (MH) is a popular MCMC algorithm which  generates candidate iterates from a proposal distribution. The algorithm then accepts or rejects candidates with a probability calibrated to ensure the process' stationary distribution coincides with the desired target posterior.  For high-dimensional target distributions, uninformed MCMC proposals like the random-walk Metropolis struggle to efficiently explore high probability areas of the posterior. By contrast, informed gradient-based MCMC schemes utilize the local geometry of a differentiable target density to steer proposals towards areas of high probability. These strategies have proven to be quite effective. Hamiltonian Monte Carlo (HMC) \cite{hanson2001markov, neal2011mcmc}, in particular, has garnered widespread popularity for Bayesian neural networks \cite{izmailov2021bayesian}, for solving complex image inverse problems \cite{altmann2014unsupervised}, and for importance sampling \cite{mousavi2021hamiltonian, LLORENTE2022310}, to name a few research areas.  Differentiability of the target density is a critical requirement for HMC and other gradient-based MCMC schemes. Bayesian models for recovering structured sparsity, however, employ non-differentiable priors by design, making it challenging to construct efficient sampling strategies.

To address this gap, \cite{pereyra2016proximal} proposed a proximal version of the gradient-based Metropolis-adjusted Langevin algorithm (MALA) \cite{roberts1996exponential}. The idea is to generate proposals using a smooth approximation of a non-smooth target via the Moreau-Yosida (MY) envelope \cite{MR144188}. 
% The MY-envelope of a convex function is a smooth approximation that bounds the function from below \cite{MR1491362}. 
 That is, proximal MCMC algorithms employ the gradient of the smooth MY-envelope in the proposal distribution, enabling the learning of local geometry despite non-smoothness. This recent success of MY-envelopes and their related proximal maps in MCMC is due to the rich literature available in convex optimization, that allows simple formulations of proximal MCMC algorithms. 
 % Applying an MH accept-reject step ensures invariance of the algorithm to the target distribution.
 % Designed initially to solve optimization problems, MY-envelopes have now been successfully employed, in conjunction with MCMC techniques, to yield efficient algorithms to sample from non-differentiable target distributions. These proximal MCMC methods use gradient of the smooth MY-envelope while proposing the next state of the chain. Pereyra \cite{pereyra2016proximal} propose the proximal versions of both the Unadjusted Langevin Algorithms (P-ULA) and the Metropolis Adjusted Langevin Algorithms (P-MALA). Both methods use the gradient of the MY-envelope in the discretised Langevin setup with the latter accepting the proposed value according to the Metropolis-Hastings accept-reject step. Later, \cite{durmus2022proximal} propose the inexact Moreau-Yosida Unadjusted Langevin algorithm (MYULA) and provide conditions for convergence of the chain as well as bounds on its convergence rate.
 % Recently, \cite{chaari2016hamiltonian} propose a non-smooth HMC (ns-HMC) which uses proximal mappings within Hamiltonian dynamics and demonstrate its applicability to various problems of image inversion and signal denoising. The ns-HMC becomes computationally expensive if the proximal mapping is not available analytically, which is the case for most problems. Furthermore, the performance worsens if the distribution is high-dimensional.

In the spirit of \cite{pereyra2016proximal},\cite{chaari2016hamiltonian} proposed non-smooth HMC (ns-HMC), utilizing gradients from a smooth MY-envelope in HMC.  
% Specifically, they replace the gradient of the potential in the HMC proposal by the gradient of its MY approximation. 
A key reason prohibiting the wide adoption of ns-HMC is the computational cost. ns-HMC becomes computationally expensive if the proximal mapping of the log target density is not available analytically and must be computed via an iterative optimization algorithm. Whereas proximal MALA only requires two calls to the approximate gradient, ns-HMC  requires multiple calls per iteration. 
% Unfortunately,  many Bayesian signal processing models yield targets with intractable MY-envelopes. 
A second drawback is that the acceptance rate of the ns-HMC algorithm is detrimentally affected by the approximation of the gradient, leading to inefficient movement in the parameter space.

In light of these drawbacks, we study an alternative HMC algorithm to handle non-differentiable targets -- proximal HMC (p-HMC) that has been employed recently in \cite{guo2023above} for specific use in altimetry. While p-HMC showed promise empirically, its theoretical properties remained unexplored. In this paper, we revisit p-HMC, establish key theoretical properties, and provide guidance  for the selection of tuning parameters. Similar to \cite{durmus2022proximal}, p-HMC splits the log target density into smooth and non-smooth components, 
and approximates only the non-smooth component by its MY-envelope. This simple difference from ns-HMC yields three key advantages.
\begin{itemize}
    \item Closed-form MY-envelopes of many non-smooth functions are abundantly available \cite{parikh2014proximal} avoiding  expensive iterative solvers for calculating approximate gradients. 
    % This opens the door to efficient implementation of an approximate leapfrog integrator.
    \item Since only part of the log target is approximated, the resulting gradient approximation is better, leading to improved mixing.
    \item It is applicable to non-convex negative log targets as long as the non-smooth component is convex. This generalizes the applicability to non-convex negative log likelihoods with sparsity inducing log-concave priors \cite{wright1999bayesian}.
\end{itemize}

\noindent\textbf{Contributions.} We establish theoretical properties of the p-HMC algorithm and evaluate its effectiveness in practice. Based on our theory and experiments, we provide practical guidance on its implementation. 
\begin{itemize}
    \item We show that the proposed p-HMC algorithm converges to the target distribution ensuring the desired MCMC asymptotics.
    \item We provide verifiable conditions for geometric ergodicity of the p-HMC sampler. Geometric ergodicity ensures that Monte Carlo estimates of means and credible intervals have finite variance, enabling reliable simulation.
    \item We provide guidelines for choosing the regularization parameter in the MY-envelope threading the trade-off between smoothness and  accuracy of the HMC leapfrog integrator.
    \item We demonstrate that p-HMC is significantly faster, easier to compute, and more effective in comparison to existing alternatives in illustrative case studies. %specific structure of the proximal mappings and the same has been demonstrated by various examples.
\end{itemize}

The remainder of the paper is organized as follows. Section~\ref{sec:background} revisits the Bayesian model, reviews HMC, and surveys the literature on existing proximal MCMC methods for non-differentiable targets. Section~\ref{sec: prox_HMC} presents our proposed methodology. Section~\ref{sec:properties} presents theoretical properties of the proposed p-HMC algorithm. Section~\ref{sec:optimal_lambda} elaborates on the choice of an appropriate value of the regularization parameter in the MY-envelope. Section \ref{sec:examples} presents examples and demonstrates the superiority of our method over ns-HMC and other samplers. Finally, Section \ref{sec:conclusion} presents concluding remarks.

\section{Preliminaries} \label{sec:background}

\subsection{Bayesian Models}
Bayesian models are often used to model signal processing data as a way to systematically incorporate prior beliefs into signal recovery \cite{MR3524962}. Leveraging prior beliefs in this way turns poorly posed or ill-posed inverse problems into well-posed ones. For example, given an observed signal $y \in \Real^n$, a linear operator $B : \Real^{d} \mapsto \Real^n$, and a possibly nonlinear operator $\xi : \Real^{n} \mapsto \Real^n$, we seek to estimate the latent signal $x \in \Real^d$ through the model
\begin{equation} \label{eq: model_defn}
    y \amp = \amp \xi(Bx)\,.
\end{equation}
Model \eqref{eq: model_defn} describes a general class of signal processing problems. For example, $x$ may be an image that we wish to recover after corruption by a blurring filter $B$ and some noise process described by $\xi$, e.g., additive noise, multiplicative noise \cite{aubert2008variational}, or impulsive noise models \cite{cai2010fast}. To estimate $x$ in \eqref{eq: model_defn}, we seek the $x$ that minimizes some measure of discrepancy between $y$ and $\xi(Bx)$. For instance, if $\xi$ describes an additive Gaussian noise model
and $B$ is full rank, then the unique minimizing $x$ is the least squares solution $\hat{x} = \arg \min_{x \in \Real^d}\|y - Bx\|^2$; here $\|\cdot\|$ denotes Euclidean norm. However, such problems can be ill-posed, lacking uniqueness or even existence of solutions \cite{marnissi2017bayesian} when $B$ is not full rank as is often the case in modern signal processing problems. Bayesian models provide an estimation framework that can  circumvent these issues.

In Bayesian models, the latent signal $x$ is a random variable that follows a prior distribution with density $\pi_0(x)$. Given the data $y$, a model or likelihood $p(y | x)$ that specifies the relationship between $x$ and $y$, we seek a posterior density of $x$ conditioned on $y$, denoted $\pi(x | y)$ (see \cite{MR2723361} for more details),
\begin{eqnarray} \label{eq: post_eqn}
    \pi(x|y) &  \propto & p(y|x)\pi_0(x)\,.
\end{eqnarray}
The posterior can be interpreted as an update on the prior of $x$ using the information encoded in the data $y$ via the likelihood. Posteriors distributions can be set up quite generally and include, but are not restricted to, likelihoods associated with \eqref{eq: model_defn}. Often it is convenient to express a posterior density as 
\begin{eqnarray} \label{eq: model_form}
    \pi(x) & \propto & \exp(-U(x))\,,
\end{eqnarray}
where $U(x) := \log \pi(x | y)$ is called the potential function. Henceforth, with a slight abuse of notation, we will denote both the target density and distribution by $\pi$. Furthermore, for notational simplicity, we denote $\pi(x |y)$ as $\pi(x)$, dropping the  posterior's dependence on the data $y$.

An additional advantage of Bayesian methodology is that it enables rigorous statistical analysis via posterior summaries. The posterior mean and the maximum a posteriori (MAP) estimate provide point estimates of the latent signal $x$. Moreover, credible intervals provide quantification of posterior uncertainty in the parameters.  While sampling may not be required to obtain a MAP estimate, it is indispensable for obtaining posterior means and constructing credible intervals.

% In order to sample the posterior distribution, the density \eqref{eq: model_form} should be completely known, that is, the normalising constant is available analytically. If this is not the case, most exact sampling techniques become unemployable. In such cases, resorting to MCMC methods is helpful owing to their ability to sample from high-dimensional target distributions known only up to proportionality. 

MCMC algorithms provide a way to sample from \eqref{eq: model_form} even when the normalizing constant is unknown. We now review MCMC basics. Let $\mathcal{B}(\Real^d)$ be the Borel sets of $\Real^d$. We  generate the iterates of discrete-time time-homogeneous Markov chain $\{X_t\}_{t\geq 1}$ with a Markov transition kernel $P: \Real^d \times \mathcal{B}(\Real^d) \mapsto [0,1]$. The kernel $P$ is the conditional distribution for obtaining $X_{t+1}$ given $X_t = x \in \Real^d$, i.e.,
\begin{equation*}
    P(x, A) \amp = \amp \text{Pr}(X_{t+1} \in A|X_{t} = x)\,,
\end{equation*}
for all $t \geq 1$ and events $A$ in $\mathcal{B}(\Real^d)$. MCMC algorithms construct kernel $P$ such that the posterior $\pi$ is its stationary distribution. If the Markov kernel $P$ is irreducible, aperiodic, and Harris recurrent, the $n$-step transition kernel converges to $\pi$ as $n \to \infty$, i.e.,
\begin{equation} \label{eq:erg_result}
   \lim_{n \to \infty} \|P^n(x, \cdot) - \pi(\cdot)\|_{\text{TV}} \amp = \amp 0 \,,
\end{equation}
where $\lVert\cdot \rVert_{\text{TV}}$ denotes the total variation distance (see \cite{meyn:twee:2009} for definitions). The MH algorithm is a powerful method to construct such Markov kernels. The class of MH algorithms is broad and  includes the random-walk Metropolis \cite{metr:1953}, MALA \cite{roberts1996exponential}, and HMC \cite{neal2011mcmc}. Of these, HMC is widely acknowledged to be the most effective algorithm, particularly suited for convex potentials.
\subsection{Hamiltonian Monte Carlo}

HMC is a popular MCMC algorithm, often exhibiting superior mixing compared to other alternatives \cite{beskos2013optimal,MR4003576}. We briefly overview HMC; for a more in-depth treatment, see \cite{neal2011mcmc}.

HMC is based on the equations of motion of particles in a system \cite{Duane1987216}.  At any given time $t$, the motion of a particle is described by its position $x \in \Real^d$ and momentum $p \in \Real^d$. The Hamiltonian, or total energy of the particle, is
% Given the target \eqref{eq: target_dens}, HMC targets the joint posterior,
% \begin{equation} \label{eq:hmc_joint_target}
%     f(x, p) \amp \propto \amp e^{-U(x)}e^{-K(p)}\,,
% \end{equation}
% where $f(x, p) = \pi(x)\phi(p)$, $\phi(p) := N(0, M_d)$. This implies $K(p) = p\Tra M_{d}^{-1}p/2$. The Hamiltonian for \eqref{eq:hmc_joint_target} is then written as,
\begin{equation} \label{eq:hamiltonian}
   H(x, p) \amp = \amp U(x) + K(p)\,,
\end{equation}
where $U(x)$ and $K(p)$ are the particle's potential and kinetic energies, respectively. The kinetic energy is defined to be $K(p) = p\Tra M^{-1}p/2$, where $M$ denotes the mass matrix of the particle.  Hamilton's equations of motion govern the particle's state in the $(x, p)$ space:
\begin{align} \label{eq: Hamilton_eq}
    \frac{dp^{(i)}}{dt} \amp = \amp \frac{\partial H(x, p)}{\partial x^{(i)}} \quad \text{and} \quad
    \frac{dx^{(i)}}{dt} \amp = \amp -\frac{\partial H(x, p)}{\partial p^{(i)}}\,,
\end{align}
for each component $i = 1, 2, \ldots, d$. If $(x_t, p_t)$ denotes the particle's state at time $t$, then the equations in \eqref{eq: Hamilton_eq} define a Hamiltonian flow operator that takes the particle's state at time $t$ to its state at time $t + \delta$. This is described by a mapping $T_\delta$ such that $T_{\delta}(x_t, p_t) = (x_{t+\delta}, p_{t+\delta})$. The solution to \eqref{eq: Hamilton_eq} defines the Hamiltonian path of constant energy. The movement along this path conserves the particle's total energy and describes its dynamics exactly. In the context of HMC, this is akin to sampling from the joint distribution with density,
\begin{equation} \label{eq: joint_hmc_targ}
    \omega(x, p) \amp \propto \amp \exp(-U(x))\exp\left(-\frac{p\Tra M^{-1}p}{2}\right)\,.
\end{equation}
Marginalizing $\omega(x,p)$ over $p$ yields  density $\pi(x)$.

If we can solve \eqref{eq: Hamilton_eq} exactly, we can sample from $\omega(x, p)$. In practice, however, analytical solutions to \eqref{eq: Hamilton_eq} are rare. Therefore,  \eqref{eq: Hamilton_eq} is numerically approximated by discretizing time at a small step size $\varepsilon$. A variety of symplectic integrators exist for constructing transformation mappings that approximately conserve the Hamiltonian. The most commonly used one is  the leapfrog integrator, which is based on the following $\varepsilon$ transition given a point $(x_0, p_0)$,
\begin{eqnarray}
    p_{\frac{\varepsilon}{2}} & = & p_{0} - \frac{\varepsilon}{2}\nabla U(x_{0}) \label{eq:leap_mom_half}\\
    x_{\varepsilon} &= &x_{0} + \varepsilon M^{-1}p_{\frac{\varepsilon}{2}} \label{eq:leap_pos_full}\\
    p_{\varepsilon} & = & p_{\frac{\varepsilon}{2}} - \frac{\varepsilon}{2}\nabla U(x_{\varepsilon})\,. \label{eq:leap_mom_last}
\end{eqnarray}
The integrator, henceforth denoted by the mapping $\tilde{T}_{\varepsilon, L}$, generates a proposal value $(x_{L\varepsilon}, p_{L\varepsilon})$ by taking these three  updates over $L$ ``leapfrog" steps. That is, $\tilde{T}_{\varepsilon, L} \approx T_{\varepsilon L}$. The sampled point is then accepted according to the Metropolis-Hastings-Green acceptance probability,
\begin{equation} \label{eq: HMC_MH_prob}
    \alpha(x_{0}, x_{L\varepsilon}) \amp = \amp \min\left\{1 , e^{-H(x_{L\varepsilon}, p_{L\varepsilon}) + H(x_{0}, p_{0})}\right\}\,.
\end{equation}
Note that as $\varepsilon \to 0$, the proposed point converges to the true solution implying that $\alpha(x_{0}, x_{L\varepsilon}) \to 1$. 

HMC and its variants \cite{devlin2024no} are common samplers employed in Bayesian signal processing as well as other Bayesian problems. Their popularity is in no small part due to their fast mixing in high-dimensions; HMC scales roughly as $O(d^{1/4})$, a significant improvement over other gradient-based schemes \cite{beskos2013optimal}. HMC and other gradient-based schemes, however, require $U(x)$ to be differentiable -- a non-trivial inconvenience for non-differentiable target densities. 
%In the sequel, we review proximal MCMC methods that employ local geometry via proximal mappings and MY-envelopes. 
% Consequently, the number of iterations HMC requires to converge increases at roughly one-fourth the corresponding increase in dimension.

\subsection{Proximal methods in MCMC}
To apply gradient-based MCMC schemes to non-differentiable targets, we turn to proximal mappings and MY-envelopes. These building blocks enable us to infer the local geometry of non-differentiable targets without requiring differentiability.

 Let $\Gamma\left(\Real^d\right)$ denote the set of convex, lower-semicontinuous, and proper functions that map $\Real^d$ into $\Real \cup \{\infty\}$. 
 %We consider problems of the form \eqref{eq: model_form}, where $U:\Real^d \to (-\infty, \infty]$ denotes the potential function.
 %In a way akin to \cite{durmus2022proximal} we assume the form $U = f + g$ where $f: \Real^d \to \Real$ and $g: \Real^d \to (-\infty, \infty]$ are lower bounded functions such that,
% \begin{enumerate}
%     \item $f$ is a smooth Lipschitz differentiable function with Lipschitz constant $C_f > 0$, i.e for all $x, y \in \Real^d$,
%     \begin{equation*}
%         \|\nabla f(x) - \nabla f(y)\| \amp \leq \amp C_f\|x - y\|\,,
%     \end{equation*}
%     \item $g$ is a convex, lower-semicontinuous, and proper function\,.
% \end{enumerate}

% Note that the function $f$ above is not required to be convex unlike \cite{durmus2022proximal}, either for the working or the ergodicity results of our algorithm. This generalises the applicability of our methodology to problems with potentials having non-convex smooth Lipschitz differentiable component.

% A large number of models constitute the class of the aforementioned problems. For instance, for $x \in \Real$ and $\gamma, \alpha > 0$, consider the model,
% \begin{equation} \label{eq: example_appln}
%     \pi(x) \amp \propto \amp \exp(-(\gamma x^2 + \alpha |x|))\,.
% \end{equation}

%We first review some optimization terminology.

\begin{definition}\label{def:moreau-envelope}
Given a positive scaling parameter $\lambda$, the Moreau-Yosida envelope of $\psi \in \Gamma\left(\Real^d\right)$ is 
\begin{eqnarray}
\label{eq:psi-lambda}
\psi^\lambda(x) & = & \inf_{y \in \Real^d} \left\{\psi(y) + \frac{1}{2\lambda}\|y-x\|^2\right\}\,.
\end{eqnarray}
\end{definition}
The infimum in \eqref{eq:psi-lambda} is uniquely attained since $\psi \in \Gamma\left(\Real^d\right)$. The minimizer of \eqref{eq:psi-lambda} defines the proximal mapping of $\psi$.

\begin{definition}\label{def:prox-operator}
Given a positive scaling parameter $\lambda$, the proximal mapping of $\psi \in \Gamma\left(\Real^d\right)$ is the operator
\begin{eqnarray}
    \label{eq:proximal_operator}
\prox_\psi^\lambda(x) & = & \underset{y \in \Real^d}{\arg\min}\; \left\{\psi(y) + \frac{1}{2\lambda}\|y-x\|^2\right\}\,.
\end{eqnarray}
\end{definition}
Proximal mappings are fundamental computational primitives in optimization algorithms \cite{Combettes2011}.
%
% Following \cite{durmus2022proximal} we define the MY-envelope of the non-smooth part $g$ as,
% \begin{equation*}
%     g^\lambda(x) \amp = \amp \inf_{y \in \Real^d} \left\{g(y) + \frac{1}{2\lambda}\|y-x\|^2\right\}\,.
% \end{equation*}
%
% \noindent Equations \eqref{eq:psi-lambda} and \eqref{eq:proximal_operator} together tell us that 
% the MY-envelope of $\psi$ can be evaluated using the proximal mapping of $\psi$,
% \begin{equation*}
%     \psi^\lambda(x) \amp = \amp \psi\left(\prox_\psi^\lambda(x)\right) + \frac{1}{2\lambda}\|x-\prox_\psi^\lambda(x)\|^2\,.
% \end{equation*}
Equation \eqref{eq:proximal_operator} can be used in \eqref{eq:psi-lambda} to evaluate the expression of $\psi^{\lambda}$ explicitly. The MY-envelope of $\psi$ has several useful properties: i) 
 $\psi^{\lambda}$ is convex whenever $\psi$ is convex and ii) is continuously differentiable even if $\psi$ is not \cite{MR1491362}. 
 The gradient of $\psi^{\lambda}$ can be expressed in terms of the proximal mapping of $\psi$ \cite{parikh2014proximal}, 
 \begin{equation} \label{eq: grad_MY_env}
    \nabla \psi^{\lambda}(x) \amp = \amp \frac{1}{\lambda}\left(x - \prox_\psi^\lambda(x)\right)\,.
\end{equation}
 Furthermore, $\nabla \psi^{\lambda}$ is $1/\lambda$-Lipschitz and
% \begin{equation*}
%     \left\|\nabla \psi^{\lambda}(x) - \nabla \psi^{\lambda}(y)\right\| \amp \leq \amp \frac{1}{\lambda}\left\|x - y\right\|\, \quad 
% \text{for all $x, y \in \Real^d$}.
% \end{equation*}
 $\psi^{\lambda}$ converges to $\psi$ pointwise \cite{MR1491362}, i.e., for all $x \in \Real^d$, $    \lim_{\lambda \to 0} \psi^{\lambda}(x) = \psi(x)$.
Thus, $\psi^{\lambda}(x)$ can smoothly approximate $\psi(x)$ arbitrarily well provided that $\lambda$ is sufficiently small. 

Proximal MCMC methods approximate non-smooth potential functions with MY-envelopes to leverage state-of-the-art gradient-based sampling schemes. Such schemes not only require differentiability but also Lipschitz differentiability to ensure their efficiency and reliability.  
% Therefore, $g^{\lambda}$ defines a smooth approximation to $g$. Using this property, an approximation to $U$ can be defined as,
% \begin{equation} \label{eq: approx_U}
%     U^{\lambda_{g}}(x) \amp = \amp f(x) + g^{\lambda}(x)\,.
% \end{equation}
% This implies for all $\lambda > 0$,
%         \begin{equation} \label{eq: target_MY_env}
%             \nabla U^{\lambda_{g}}(x) = \nabla f(x) + \frac{1}{\lambda}\left(x - \prox_g^\lambda(x)\right)\,.
%         \end{equation}
%         In addition, $\nabla U^{\lambda_{g}}$ is Lipschitz with constant $C \leq C_f + \lambda^{-1}$.
% 
% Recently, they have also been employed within the MCMC framework to obtain samples from distributions with non-differentiable density functions. This has given rise to proximal MCMC methods.
% 
For example, \cite{pereyra2016proximal} approximates a target density $\pi$ with potential $U \in \Gamma\left(\Real^d\right)$ using 
\begin{equation*}
    \pi^{\lambda}(x) \amp \propto \amp \exp\left(-U^{\lambda}(x)\right)\,, 
\end{equation*}
where $U^{\lambda}$ is the MY-envelope of $U$. They further employ gradients of $U^{\lambda}$ in a MALA-like algorithm yielding proximal Metropolis-adjusted Langevin algorithm (p-MALA). Specifically,  given the current state $x_t$, the next proposed state is 
\begin{equation} \label{eq:ULA_prop}
     x^* \amp = \amp x_t - \frac{h}{2}\nabla U^{\lambda}(x_t) + \sqrt{h}\zeta\,,
\end{equation}
where $\zeta \sim N(0, \mathbb{I}_d)$ and $h$ is a step size parameter. 
% In p-ULA, the proposed state is always accepted, i.e., $x_{t+1} = x^*$. The p-ULA process converges to an approximation of $\pi^{\lambda}$. To generate samples from $\pi$, Pereyra proposes p-MALA. In p-MALA, 
The proposal \eqref{eq:ULA_prop} is accepted, i.e., $x_{t+1} = x^*$,  according to the MH accept-reject ratio
\begin{equation*}
    \alpha(x_t, x^*) \amp = \amp \min \left\{1, \frac{\pi(x^*)q(x^*, x_t)}{\pi(x_t)q(x_t, x^*)}\right\}\,,
\end{equation*}
where $q(x_t, x^*)$ is the proposal density for \eqref{eq:ULA_prop}. The choice of $\lambda$ in p-MALA is fixed to be $\lambda = h/2$.

MY-envelopes have also been proposed for HMC in \cite{chaari2016hamiltonian}, who presented the non-smooth HMC (ns-HMC). The ns-HMC employs $\nabla U^{\lambda}(x)$ in the leapfrog integrator of HMC, fixing $\lambda = 1$. The ns-HMC algorithm uses the following modifications to \eqref{eq:leap_mom_half}, \eqref{eq:leap_pos_full}, \eqref{eq:leap_mom_last} in the standard HMC algorithm
\begin{eqnarray} 
    p_{\frac{\varepsilon}{2}} & = & p_{0} - \frac{\varepsilon}{2}\left[x_0 - \prox_{U}^{1}(x_0)\right] \label{eq: ns_hmc_eqns_p_first} \\
    x_{\varepsilon} &= &x_{0} + \varepsilon p_{\frac{\varepsilon}{2}}  \label{eq: ns_hmc_eqns_x} \\
    p_{\varepsilon} & = & p_{\frac{\varepsilon}{2}} - \frac{\varepsilon}{2}\left[x_{\varepsilon} - \prox_{U}^{1}(x_{\varepsilon})\right]\,,  \label{eq: ns_hmc_eqns_p_second}
\end{eqnarray}
where $\prox_{U}^{1}$ denotes the proximal mapping of the potential $U$ for $\lambda = 1$. The steps in \eqref{eq: ns_hmc_eqns_p_first}-\eqref{eq: ns_hmc_eqns_p_second} are repeated $L$ times to obtain the proposed value $(x_{L\varepsilon}, p_{L\varepsilon})$.
% The scheme proposes to substitute $\nabla U$ in \eqref{eq: ns_hmc_eqns_p_first} and \eqref{eq: ns_hmc_eqns_p_second} with its subgradient at the point $\prox_{U}^{1}$. 
% Note that from \eqref{eq: grad_MY_env}, this replacement is $\nabla U^{\lambda}$ for $\lambda = 1$. 
The proposal is then accepted or rejected according to \eqref{eq: HMC_MH_prob}. The ns-HMC opens the door to using HMC-like algorithms to sample from non-differentiable targets. It, however, has some drawbacks.
\begin{itemize}
    \item The algorithm inherently assumes $\lambda = 1$ and does not explore the effect of varying $\lambda$ on the performance of the algorithm.
    
    %\item The proposed algorithm fixes $\lambda = 1$, but does not investigate how this value might affect the performance of the sample.
    
    \item In Bayesian models, the likelihood is often complex enough that the proximal mapping of the full potential, $U$, is not available analytically. Thus, generating each proposal of ns-HMC requires $L$ repeated calls to an iterative solver, i.e., a subroutine based on proximal splitting techniques like the forward-backward  algorithm, Douglas-Rachford splitting algorithm, Dykstra-like splitting method, etc (see \cite{Combettes2011} for details). This adds a large computational burden.
    % for one run of the chain, adding heavily to the computational burden of the algorithm. Note that the ns-HMC algorithm requires evaluation of the proximal mapping for the full potential $U$. This is tedious for most problems and This coupled with the fact that the leapfrog is run for $L$ steps increases the cost of computation multi-fold. 

    \item The convexity of $U$ is a critical assumption and thus the sampler is not applicable to target densities that may not be log-concave.

    \item Lastly, often $U$ is composed of a smooth and a non-smooth part. Enveloping all of $U$ approximates both these parts leading to a cruder approximation than possibly necessary. 
\end{itemize}
This last point was recognized by \cite{durmus2022proximal} who propose approximating only the non-smooth part of $U$, yielding another alternative called Moreau-Yosida MALA (my-MALA) algorithm. 
% Later, \cite{durmus2022proximal} propose the Moreau-Yosida ULA (myULA) and the Moreau-Yosida MALA (myMALA) algorithm. Assuming the integrability of $\exp(-U(x))$, they prove that $\exp(-U^{\lambda}(x))$ is also integrable, concluding that $\pi^{\lambda}$ is a valid probability density. Both myULA and myMALA target the approximated density $\pi^{\lambda}$ and explicit bounds on the rate of geometric convergence have also been established \cite{durmus2022proximal}.

% Since its inception in the context of MCMC in \cite{neal2011mcmc}, HMC has been used extensively across various domains. It is most successfully employed for high-dimensional Bayesian models. This is because HMC scales as $O(d^{1/4})$ for large dimensions \cite{beskos2013optimal}, that is, the increase in the number of steps required for convergence is roughly of order one-fourth of that of the dimensions. This is better than most other gradient-based algorithms like MALA and Barker's where the scaling is $O(d^{1/3})$. %Tuning free HMC methods and their variants have also been proposed and have garnered widespread acceptance \cite{MR3214779, hoffman2022tuning}.

\section{Proximal Hamiltonian Monte Carlo} \label{sec: prox_HMC}

% In order to overcome the drawbacks of ns-HMC we propose a different modification of the leapfrog integrator. 
In this section, we adapt the methodology of \cite{durmus2022proximal} to address drawbacks of the ns-HMC algorithm. 
Let $U:\Real^d \to \Real \cup \{\infty\}$ be a target potential in \eqref{eq: model_form} that can be expressed as 
\[
U(x) \amp = \amp f(x) + g(x)\,,
\]
where 
\begin{enumerate}
    \item $f: \Real^d \to \Real$ is a Lipschitz differentiable function with Lipschitz constant $C_f > 0$, i.e., for all $x, x' \in \Real^d$,
    \begin{equation*}
        \|\nabla f(x) - \nabla f(x')\| \amp \leq \amp C_f\|x - x'\|\,,
    \end{equation*}
    \item $g:\Real^d \to \Real \cup \{\infty\}$ is a convex, lower-semicontinuous, and proper function\,.
\end{enumerate}
Note that $f$ is not required to be convex unlike in \cite{chaari2016hamiltonian,durmus2022proximal}, either for the implementation or the ergodicity results of our algorithm. Thus, $U$ may be non-convex as long as $g$ is convex and $f$ is Lipschitz differentiable. 
% This generalises the applicability of our methodology to problems with potentials having non-convex smooth Lipschitz differentiable component.
%
For our method, we propose using gradients of the following approximation:
\begin{equation} \label{eq:potential_split}
    U^{\lambda_g}(x) \amp := \amp f(x) + g^{\lambda_g}(x)\,,
\end{equation}
where $g^{\lambda_g}$ denotes the MY-envelope of $g$ with parameter $\lambda_g$. We propose employing a proximal HMC (p-HMC) algorithm where each leapfrog step is
\begin{eqnarray}
    p_{\frac{\varepsilon}{2}} & = & p_{0} - \frac{\varepsilon}{2}\nabla U^{\lambda_{g}}(x_{0}) \label{eq: P_HMC_p_first}\\
    x_{\varepsilon} &= &x_{0} + \varepsilon p_{\frac{\varepsilon}{2}} \label{eq: P_HMC_x} \\
    p_{\varepsilon} & = & p_{\frac{\varepsilon}{2}} - \frac{\varepsilon}{2}\nabla U^{\lambda_{g}}(x_{\varepsilon})\,. \label{eq: P_HMC_p_second}
\end{eqnarray}
A total of $L$ steps are performed to get to the proposed value, that is, if $\tilde{T}^{\lambda_g}_{\varepsilon, L}$ denotes the deterministic leapfrog transition, then $\tilde{T}^{\lambda_g}_{\varepsilon, L}(x_0, p_0) = (x_{L\varepsilon}, p_{L\varepsilon})$ is the proposed value. Subsequently, this is accepted with probability,
\begin{eqnarray*}
        \alpha(x_{0}, x_{L\varepsilon}) & = & \min\left\{1 , e^{-H(x_{L\varepsilon}, p_{L\varepsilon}) + H(x_{0}, p_{0})}\right\}\,.
    \end{eqnarray*}
Algorithm~\ref{algo:p-HMC} provides pseudocode for the p-HMC sampler.
% Equations \eqref{eq: P_HMC_p_first}-\eqref{eq: P_HMC_p_second} approximate $\nabla U$ by $\nabla U^{\lambda_g}$. This is reasonable since $U^{\lambda_g}$ is close to $U$ for $\lambda$ being small.

\begin{algorithm} [t]
\caption{p-HMC given $\varepsilon > 0$, $L \geq 1$ and $M \succ 0$}
\label{algo:p-HMC}
    \begin{enumerate}
    \item Given $X_{t} = x_{0}$, draw $p_{0} \sim N(0, M)$
    \item Use the leapfrog one-$\varepsilon$ step equations,
    \begin{eqnarray*}
    p_{\frac{\varepsilon}{2}} & = & p_{0} - \frac{\varepsilon}{2}\nabla U^{\lambda_{g}}(x_{0}) \\
    x_{\varepsilon} &= &x_{0} + \varepsilon M^{-1}p_{\frac{\varepsilon}{2}} \\
    p_{\varepsilon} & = & p_{\frac{\varepsilon}{2}} - \frac{\varepsilon}{2}\nabla U^{\lambda_{g}}(x_{\varepsilon})
\end{eqnarray*}
$L$ times sequentially to reach $(x_{0}, p_{0}) \rightarrow (x_{L\varepsilon}, p_{L\varepsilon})$
    \item Draw \(U \sim U(0,1)\) independently, and calculate
    \begin{eqnarray*}
        \alpha(x_{0}, x_{L\varepsilon}) & = & \min\left\{1 , e^{-H(x_{L\varepsilon}, p_{L\varepsilon}) + H(x_{0}, p_{0})}\right\}\,,
    \end{eqnarray*}
    where $H(x, p) = U(x) + \frac{1}{2}p^{\Tra}M^{-1}p$ denotes total energy at $(x, p)$
    \item If $U \leq \alpha(x_{0}, x_{L\varepsilon})$, then $X_{t+1} = x_{L\varepsilon}$
    \item Else $X_{t+1} = x_0$
\end{enumerate}
\end{algorithm}

% More specifically, if $n_{\epsilon}$ denotes the number of steps required to be within $\epsilon$ distance of proximal mapping by the iterative algorithm, then the total steps to propose a state becomes $n_{\epsilon}L$.

The proposed p-HMC algorithm often admits efficient implementations since proximal mappings of various $g$ functions are available analytically \cite{parikh2014proximal}.
% , ``splits" the smooth and non-smooth part of $U$, and therefore can often accommodate for the availability of the closed form proximal mapping of the function $g$. 
% 
% A large number of models constitute the class of the aforementioned problems. 
For instance, for $x \in \Real^d$ and $\gamma, \alpha > 0$, consider the popular structure,
\begin{equation} \label{eq: example_appln}
    \pi(x) \amp \propto \amp \exp\left\{-\left(\gamma \lVert x \rVert^2 + \alpha \lVert  x\rVert_1\right)\right\}\,.
\end{equation}
The proximal mapping for the $\ell_1$-norm has a closed form -- the soft-thresholding operator -- that can be computed in $O(d)$ flops. Thus, proposals can be generated quite cheaply since the computational cost of $\nabla U^{\lambda_{g}}$ is low. 
% Equivalently, the number of steps required to propose a state reduces to $L$ from $n_{\epsilon}L$ before. 
Moreover, even when a proximal mapping is available in closed form for both p-HMC and ns-HMC, our method performs better as is evidenced by the effective sample size  in numerical studies (Section~\ref{sec:examples}). This is because p-HMC only approximates $g$ but takes the correct gradient contribution from $f$ making  $U^{\lambda_g}$ a better approximation than $U^{\lambda}$. As a result, the Hamiltonian is better conserved thereby ensuring a reasonable acceptance rate. We revisit this point in  more detail in Section~\ref{sec:examples}.

% Unfortunately, since ns-HMC is only defined for $\lambda = 1$, there is no discussion regarding the choice of $\lambda$ or the potential use of different MY-envelopes. We address this by incorporating $\lambda$ in our methodology and also provide a discussion on a possibly appropriate choice of the same.

\section{Properties of p-HMC} \label{sec:properties}
We analyze the asymptotic behavior of the p-HMC algorithm. Section ~\ref{sec:invariance} describes the existence of the stationary and limiting distribution of the p-HMC chain. Section~\ref{sec: geom_erg} details its rate of convergence. 

\subsection{Invariance and convergence} \label{sec:invariance}
The conditions for ensuring invariance for the correct target distribution in case of HMC is more nuanced than a usual MH algorithm. % Although the limiting distribution for the MH algorithm is generally invariant to the correct target distribution, conditions for ensuring the same limiting behavior for HMC is more nuanced.
This is because HMC is composed of a deterministic proposal mechanism conditioned on a randomly sampled auxiliary variable. To navigate this, we turn to the generalised version of the MH algorithm for  deterministic proposals \cite{10.1093/biomet/82.4.711}, the Metropolis-Hastings-Green (MHG) algorithm.

The MHG algorithm consists of augmenting the state space with that of an auxiliary variable. The auxiliary variable is then chosen randomly conditioned on the initial state of the variable of interest. Subsequently, the augmented state is then mapped to a proposal deterministically which is accepted according to the MHG ratio. More specifically, let $v$ denote the auxiliary variable and $(x, v)$, the augmented space. The proposal value at any iteration $j+1$ is obtained in the following way: i) sample $v \sim q_{\text{aux}}(x_j, \cdot)$, the conditional distribution of the auxiliary variable given the previous state, and ii) employ a deterministic mapping $\phi(x_j, v) = (x^*, v^*)$ to obtain the proposed value. The resulting proposal value is then accepted with the following MHG probability, 
% In general, for a target density $\pi$,  let $q_{\text{aux}}(x, \cdot)$ denote the ``proposal" density of the auxiliary variable given a state $x$. Further, given any initial state $(x_0, p_0)$, let $\phi$ denote the mapping from the initial state to the proposed state, that is, $\phi(x_0, p_0) = (x^*,p^*)$. Then, the proposed state is accepted with the MHG probability
\begin{equation} \label{eq: MHG_ratio}
    \alpha(x_0, x^*) \amp = \amp \min \left\{1, \frac{\pi(x^*)q_{\text{aux}}(x^*, v^*)}{\pi(x_0)q_{\text{aux}}(x_0, v_0)}|\text{det}\,J_{\phi}|\right\}\,,
\end{equation}
where $J_{\phi}$ is the Jacobian of the mapping $\phi$ (see \cite{10.1093/biomet/82.4.711} for more details). In general, any algorithm employing \eqref{eq: MHG_ratio} requires evaluating $J_{\phi}$ for every iteration. However, HMC avoids this due to volume preservation. Volume preservation implies that the Jacobian of transformation in the $(x, p)$ space is unity \cite{neal2011mcmc}. Therefore, \eqref{eq: MHG_ratio} is equivalent to accepting with probability \eqref{eq: HMC_MH_prob}. 

Note that p-HMC employs the approximation $\nabla U^{\lambda_g}$ instead of $\nabla U$ and thus uses an approximate Hamiltonian in the proposal step. The proposal mechanism uses the leapfrog integrator, which is reversible and volume preserving. Therefore, using the MHG accept-reject probability in  \eqref{eq: MHG_ratio} suffices to ensure that the samples are asymptotically from $\pi$ (see \cite{neal2011mcmc} for a discussion). Consequently, an argument similar to \cite{neal2011mcmc} applies and the p-HMC algorithm is volume preserving, and thus defines a Markov kernel invariant for $\pi$.

Note that $\exp\left({-U^{\lambda_{g}}(x)}\right)$ does not need to be integrable. This eliminates the requirement of the potential $U^{\lambda_g}$ to define a distribution and makes our approach broadly applicable.

% For the methodology described here, it follows directly from \cite{chaari2016hamiltonian}. Note that,
% \begin{equation*}
%     U(x) = f(x) + g(x) \qquad \implies \qquad \partial U(x) = f'(x) + \partial g(x)
% \end{equation*}
% where the operator $\partial$ denotes the subdifferential set. Substituting this in the subdifferential of the Hamiltonian $H(x, p)$ and following the arguments of \cite{chaari2016hamiltonian} leads to the result. 

\subsection{Geometric ergodicity}
\label{sec: geom_erg}

The efficiency of an MCMC algorithm is typically characterized by its rate of convergence. In this respect, an important property is geometric ergodicity which we recall below.

\begin{definition}
    Let $P$ be an irreducible, aperiodic, and Harris recurrent Markov kernel with stationary distribution $\pi$. Then, for some $\rho < 1$ and $V: \Real^d \to \Real^+$ if,
    \begin{equation}
     \|P^{n}(x, \cdot) - \pi(\cdot)\|_{\text{TV}} \amp \leq \amp V(x)\rho^{n}\,, \qquad n = 1, 2, 3, \ldots
    \end{equation}
     for all $x \in \Real^d$, the Markov chain produced is said to be geometrically ergodic.
\end{definition}

Geometric ergodicity ensures that the Markov chain explores the state space reasonably well, ensuring also the existence of a central limit theorem (CLT). Existence of a CLT is important since it enables quantifying the uncertainty in estimators of both posterior means and quantiles \cite{MR3757517}. 
% 
% \begin{theorem}[\!\!\cite{MR3757517}]
%     Let $\{X_t\}_{t\geq1}$ be a geometrically ergodic Markov chain with stationary distribution $\pi$. Consider a Borel function $\zeta: \Real^d \to \Real^p$, and let $\Bar{\zeta}_n = n^{-1} \sum_{t=1}^n \zeta(X_t)$. If  $\E_{\pi}\|\zeta(x)\|^{2+\gamma} < \infty$ for some $\gamma > 0$, then for any initial distribution as $n \to \infty$,
%     \begin{equation*}
%         \sqrt{n}\left(\Bar{\zeta}_n - \E_{\pi}(\zeta(x)) \right) \amp {\xrightarrow[]{\text{d}}} \amp N_p(0, \Sigma_{\zeta})\,,
%     \end{equation*}
% where $\Sigma_\zeta$ is known. A similar result for the asymptotic distribution of quantiles also holds \cite{MR3285872}.
% \end{theorem}
% 
 Under suitable conditions, the original HMC algorithm has been shown to be geometrically ergodic \cite{MR4003576}. 
 
 We next establish sufficient conditions under which our p-HMC is geometrically ergodic. Similar to the the HMC proposal, the p-HMC proposal can be written as an $L$-step generalisation of the proximal Langevin scheme (see \cite{MR4003576} for details). Let $x_0$ be the current state and without loss of generality, $p_0 \sim N(0, M)$, where $M = \mathbb{I}_d$. The p-HMC proposal can then be expressed as
\begin{eqnarray} \label{eq:HMC_proposal}
    x_{L\varepsilon}  =  m_{L, \varepsilon}(x_0, p_0)  + L\varepsilon p_{0} \,,
    \end{eqnarray}
where,
\begin{equation*}
    m_{L, \varepsilon}(x_0, p_0) = x_{0} - \frac{L\varepsilon^{2}}{2}\nabla U^{\lambda_g}(x_{0}) - \varepsilon^{2}\sum_{i=1}^{L-1} (L-i)\nabla U^{\lambda_g}(x_{i\varepsilon})\,.
\end{equation*}
Furthermore, let the density of the proposal in \eqref{eq:HMC_proposal} be denoted by $q_{\text{pH}}(x_0, x_{L\epsilon})$ given the state $x_0$. The following assumption is made to ensure that the p-HMC chain is irreducible (see \cite{MR4003576} for details).
\begin{assum}
\label{ass:L}
The number of leapfrog steps $L$ is sampled from  a distribution $\mathcal{L}(\cdot)$ such that $\Pr_{\mathcal{L}}[L=1] > 0$ and for any $(x_{0}, p_{0}) \in \mathbb{R}^{2d}$ and $\varepsilon > 0$, there is an $s < \infty$ such that $\mathbb{E}_{\mathcal{L}}\left[e^{s\|x_{L\varepsilon}\|}\right] < \infty$.
\end{assum}
Let, 
\begin{equation*}
    R(x) \amp = \amp \left\{x^* : \dfrac{\pi(x^*)q_\text{pH}(x^*, x)}{\pi(x)q_\text{pH}(x, x^*)} < 1  \right\}\,
\end{equation*}
  denote the potential rejection region and $I(x)  =  \{x^* : \|x^*\| \leq \|x\|\}$ denote the set of points interior to $x$, respectively.

\begin{assum} \label{ass:living_limsup}
   Let $f$ and $g$ be such that
    \begin{enumerate}[label=(\alph*)]
        \item $\lim_{\|x\| \to \infty} \|\nabla f(x)\|  =  \infty$, 
        \item $\limsup_{\|x\| \to \infty} \dfrac{\|\nabla f(x)\|}{\|x\|}  < \infty$,
        \item $\liminf_{\|x\| \to \infty} \dfrac{\langle \nabla f(x), x \rangle}{\|\nabla f(x)\|\|x\|}  >  0$, and
        \item $\lim_{\|x\| \to \infty} \dfrac{\|\nabla g^{\lambda_g}(x)\|}{\|\nabla f(x)\|} < 1$.
     \end{enumerate}
\end{assum}
\begin{theorem} \label{thm:geom_erg}
    Under Assumptions \ref{ass:L} and \ref{ass:living_limsup}, a p-HMC ergodic Markov chain is  geometrically ergodic if,
    \begin{eqnarray} 
         \lim_{\|x\| \rightarrow \infty} \int_{R(x) \cap I(x)} q_{\text{pH}}(x, x^*) dx^* & = & 0 \,.
    \end{eqnarray}    
\end{theorem}

\begin{proof}
   See Appendix~\ref{ge_thm_proof}.
\end{proof}

\begin{remark}
     Condition (a) of Assumption \ref{ass:living_limsup} ensures that the target density function $\pi(x)$ decays in the tails whereas condition (b) states that the rate of growth in $\nabla f(x)$ is asymptotically sublinear. Condition (c) is the `inwards convergence' condition (see \cite{MR4003576}). Condition (d) relates to the specific form of penalty function $g^{\lambda_g}$ compared to $f$. It implies that the rate of growth of $\nabla g^{\lambda}$ is not faster than that of $\nabla f$. This can occur naturally as we describe next.
 \end{remark}

\subsection{Example: The $\ell_1$-norm prior}

Consider $g(x) = \|x\|_1 = \|(x_{(1)}, \dots, x_{(d)})^{\Tra}\|$. This is the popular $\ell_1$-penalty which is common in many Bayesian problems \cite{parkandcasella, 382009}. Here,
    \begin{equation}
\label{eq:softthreshold}        \left[\prox_{g}^{\lambda_g}(x)\right]_i = \begin{cases}
            x_{(i)} - \lambda_g \quad &\text{if}\,\,\, x_{(i)} > \lambda_g \\
            x_{(i)} + \lambda_g \quad &\text{if}\,\,\, x_{(i)} < \lambda_g \,\quad i = 1, 2, \ldots, d, \\
            0 &\text{if}\,\,\, |x_{(i)}| \leq \lambda_g\,
        \end{cases}
    \end{equation}
which is the element-wise soft-thresholding operator. Therefore,
    \begin{equation} \label{eq: grad_soft_l_1}
    \left[\nabla{g}^{\lambda_g}(x)\right]_i = 
    \begin{cases}
        1 \quad &\text{if } x_{(i)} > \lambda_g \\
        -1 \quad &\text{if } x_{(i)} < \lambda_g \\
        x_{(i)}/\lambda_g &\text{if } |x_{(i)}| \leq \lambda_g\,.
    \end{cases}
\end{equation}
    This implies, under Assumption \ref{ass:living_limsup}(a) that
    \begin{equation*}
        \frac{\|\nabla{g}^{\lambda_g}(x)\|}{\|\nabla{f}(x)\|} \leq \frac{\sqrt{d}}{\|\nabla{f}(x)\|} \to 0 \qquad \text{as} \qquad \|x\| \to \infty\,.
    \end{equation*}
    Thus, condition (d) holds for the $\ell_1$-norm provided (a) holds.

\section{Choice of $\lambda_g$} \label{sec:optimal_lambda}
% It is fairly evident from the above discussion that an appropriate choice of the regularisation parameter $\lambda$ is critical to the efficient working of p-HMC algorithm. The following discussion provides a heuristic idea about the choice of $\lambda$ and the intuition behind it.

The choice of the regularisation parameter is critical to the performance of any proximal MCMC algorithm. Whereas \cite{chaari2016hamiltonian} fix $\lambda = 1$ for ns-HMC, such a choice may not be universally suitable over all applications and datasets. Indeed, $\lambda_g$'s value trades-off smoothness of $U^{\lambda_g}$ and its quality of approximation to the target Hamiltonian. That is, increasing the penalty makes the potential smoother, but reduces the acceptance probability since the true Hamiltonian is less preserved. This trade-off must be balanced for effective  sampling. In this section, we provide guidance on choosing $\lambda_g$.

% Let $H^{\lambda}(x, p)$ denote the Hamiltonian corresponding to $U^{\lambda}$. Further, let the leapfrog discretization $\varepsilon$ very small ($\leq 10^{-5}$ ppopularably) and $L =1$. 

Intuitively, we should choose $\lambda_g$ so  that the approximation error in $g^{\lambda_g}$ is small compared to the information in $f$ about $U$. For instances where $f$ contributes more information to $U$, a larger discrepancy between $g$ and $g^{\lambda_g}$ might be tolerable, but for instances where $f$ contributes less to $U$ compared to $g$, the  approximation error in $g^{\lambda_g}$ must be restricted.  With this in mind, let $x_0$ be a minimizer of $U$ so that $0 \in \partial U(x_0)$, where $\partial U$ denotes the subdifferential of $U$. Note that $x_0$ is not a minimizer of $U^{\lambda_g}$ although $x_0$ will be close to a minimizer of $U^{\lambda_g}$
for a small enough $\lambda_g$. 

Let $L = 1$ and $(x_0, p_0)$ be the initial point. Consider $\varepsilon \approx 0$ so that the approximation error to the Hamiltonian is dominated by the MY-smoothing. We track the relative change in the Hamiltonian $H$ for one leapfrog step over a range of $\lambda_g$. For sufficiently large  $\lambda_g$, the approximation error in using $g^{\lambda_g}$ will be large enough to detrimentally impact the approximation to the Hamiltonian. Thus,
% Since the choice of $\varepsilon$ is very close to zero, $T_{\lambda, \varepsilon}(x_0, p_0) \approx T_{\lambda, 0}(x_0, p_0)$ where $T_{\lambda, 0}$ denotes the continuous time alternative transition of $U^{\lambda}$ as $\varepsilon \to 0$. 
consider the metric
\begin{equation} \label{eq: lambda_metric}
    R^{\lambda_g} \amp = \amp \left|\frac{H(x_0, p_0) - H\left(\tilde{T}^{\lambda_g}_{\varepsilon, L}(x_0, p_0)\right)}{H(x_0, p_0)}\right|\,.
\end{equation}
 Equation \eqref{eq: lambda_metric} measures the relative difference in the value of the Hamiltonians of the target at the initial and the proposed point, respectively. Recall that for $\varepsilon \approx 0$, $U^{\lambda_g}(x) \to U(x)$ for all $x$ as $\lambda_g \to 0$ and consequently  $H\left(\tilde{T}^{\lambda_g}_{\varepsilon, L}(x_0, p_0)\right)$ tends to $H(x_0, p_0)$ as $\lambda_g \to 0$. Thus, on the one hand, with $\varepsilon \approx 0$, the relative error $R^{\lambda_g}$ becomes very small as $\lambda_g \to 0$, ensuring high acceptance rates of the algorithm.  On the other hand, increasing $\lambda_g$ smooths the potential yielding better behaved gradients. There is thus, a trade-off between choosing a large $\lambda_g$, which ensures smoothness in the approximated potential, and a small one, which increases the acceptance rate.

\begin{figure*}
    \centering
    \includegraphics[width=0.3\linewidth]{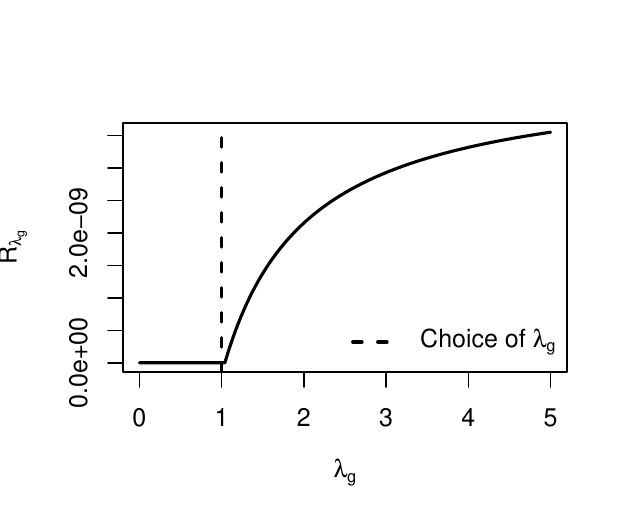}
    \includegraphics[width=0.3\linewidth]{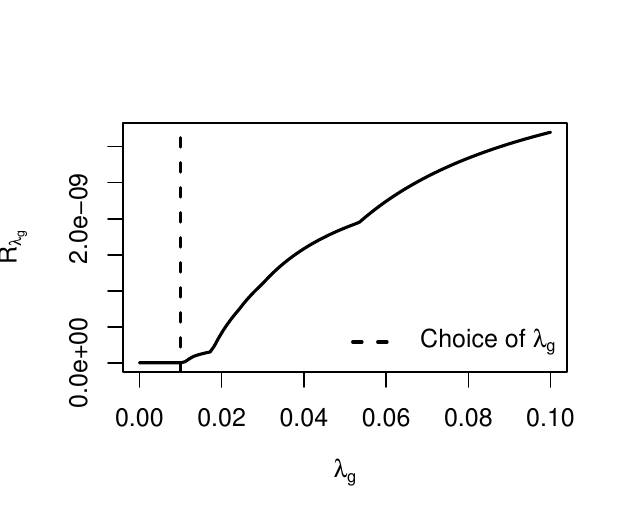}
    \includegraphics[width=0.3\linewidth]{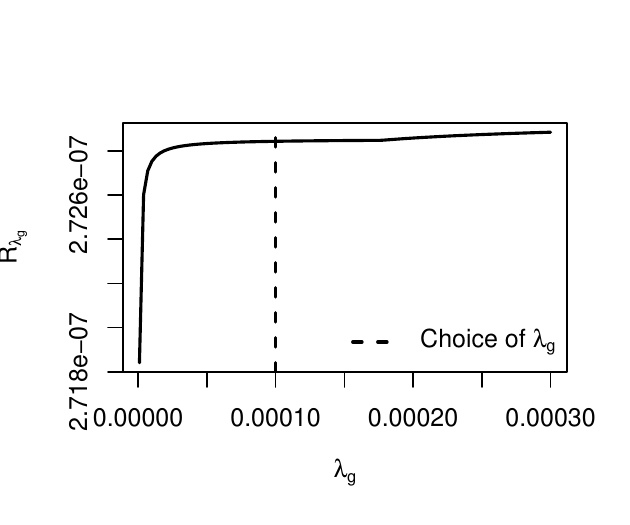}
    \caption{Choice of $\lambda_g$ for the three examples in Section~\ref{sec:examples}, with $\varepsilon = 10^{-7}$: (left) Toy example, (middle) sparse logistic regression, (right) low-rank matrix estimation.}
    \label{fig:lambda-toy}
\end{figure*}
Figure~\ref{fig:lambda-toy} shows how $R^{\lambda_g}$ varies with $\lambda_g$ for the three examples in Section~\ref{sec:examples}. In all three examples, $R^{\lambda_g}$ is nondecreasing in $\lambda_g$.
We recommend choosing the largest $\lambda_g$ so that $R^{\lambda_g}$ remains small.

\section{Numerical studies}
\label{sec:examples}
 This section demonstrates the usefulness of the p-HMC algorithm over ns-HMC and other MCMC methods. We begin with a one dimensional Bayesian lasso example which provides some validation for the proposed choice of $\lambda_g$ in the p-HMC algorithm. The next two examples, are low-dimensional signal recovery problems which illustrate the advantages of p-HMC over the alternatives in terms of mixing and estimation accuracy.
 %The next two examples, a sparse Bayesian logistic regression and a nuclear norm based low-rank matrix estimation, are canonical low-dimensional signal recovery problems. 
 %These examples illustrate the advantages of p-HMC over the alternatives in terms of mixing and estimation accuracy.
 We use $M = \mathbb{I}_{d}$ for all the implementations. The code is available at \url{https://github.com/sapratim/P-HMC}.

\subsection{Toy example} % (fold)
\label{sub:toy_example}

We first consider a simple one-dimensional problem where we compare the leapfrog discretization for both p-HMC and ns-HMC. We use this simple case study to visualize the effects of $\lambda_g$ (p-HMC) and $\lambda$ (ns-HMC) on the numerical integrator and the subsequent leapfrog trajectories' deviation from the true potential.

For $n = 100$, let $y = (y_1, y_2, \dots, y_n)$ and for $x \in \mathbb{R}$, consider the target density with potential
\begin{eqnarray}
\label{eq:toy_lasso}
	U(x)  = \dfrac{1}{2} \sum_{i=1}^{n}(y_i - x)^2 + |x|  =:  f(x) + g(x)\,.
\end{eqnarray}
For this one-dimensional Bayesian lasso-like model in \eqref{eq:toy_lasso}, we draw contours of $e^{-H(x, p)}$ to visualize how well the leapfrog steps of each algorithm hug the contour. The proximal mappings of $g$ and $U$ are available in closed form:
\begin{align}
\label{eq:proximal_toy_U}
	&\prox_{g}^{\lambda_g}(x) \amp = \amp \text{sign}(x)\cdot\max\{|x| - \lambda_g, 0\} \quad \text{and,} \\
    &\prox_U^{\lambda}(x) \amp = \amp \text{sign} (w) \cdot \max\left\{ |w| - \dfrac{\lambda}{1 + n \lambda}, 0\right\}\, \label{eq:proximal_toy_U_lambda},
\end{align}
where $w = (x + \lambda \sum_i y_i)/(1 + n \lambda)$. We generate $y_i \overset{\text{iid}}{\sim} N(1, 0.5)$. Using the proximal mappings in \eqref{eq:proximal_toy_U} and \eqref{eq:proximal_toy_U_lambda}, for fixed $\lambda_g, \lambda, \varepsilon,$ and $L$, we plot the leapfrog steps of the integrators for p-HMC and ns-HMC.  Figure~\ref{fig:contours} shows that the leapfrog steps generated by p-HMC hug the Hamiltonian contour well for an extremely wide range of $\lambda_g$ choices, whereas the leapfrog steps generated by ns-HMC are not as reliable. Specifically, ns-HMC's steps can hug the contour closely for $\lambda = 0.001$ but may go far off course for larger $\lambda$. p-HMC's ability to conserve the Hamiltonian with high fidelity across a wide range of $\lambda_g$ ensures high acceptance rates and quality of sampling that is fairly  agnostic to the choice of $\lambda_g$. 
\begin{figure*}[t]
\begin{center}
	\includegraphics[width=0.9\linewidth]{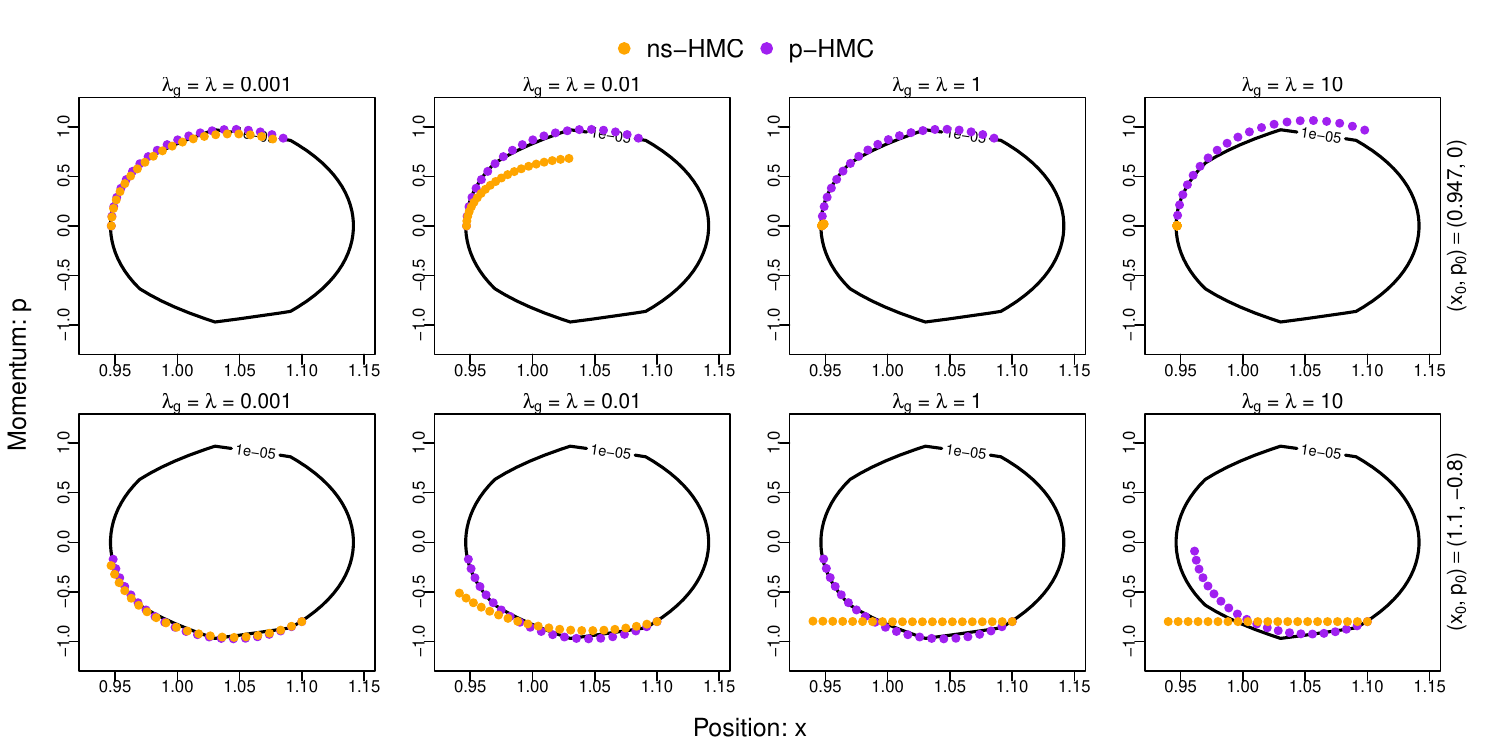}
\caption{Toy example: Leapfrog discretised step for $\varepsilon = 0.01$ and $L = 20$ with two different starting values in each row.}
\label{fig:contours}
\end{center}
\end{figure*}
% subsection toy_example (end)
\vspace{-4pt}
\subsection{Sparse logistic regression}
\label{sec:logistic}
We next consider the popular sparse logistic regression model. The data consists of  $n$ binary responses, i.e.,  $y = (y_1, \ldots, y_n)$ where each $y_i$ is zero or one, and predictor variables $x_i \in \Real^{d}$ for $i = 1, \ldots, n$. We assume the $y_i$'s are independent and  $y_i = 1$ with probability $(1 + \exp\left(-x_i\Tra\beta)\right)^{-1}$, where $\beta \in \Real^d$ is a regression parameter of interest. 
%where the binary response is modeled as,
%\begin{equation*}
%    (Y_{i} = y_i) \sim \text{Bern}\left(\frac{e^{x_{i}\Tra\beta}}{1 + e^{x_{i}\Tra\beta}}\right) \qquad \quad i = 1, 2, \ldots, n, 
%\end{equation*}
%where $x_i, \beta \in \Real^d$ denote the $i^{\text{th}}$ covariate and the parameter vectors, respectively. 
% 
To recover a sparse model, we assume an $\ell_1$-prior on $\beta$. This acts as the non-differentiable component in the posterior distribution with proximal mapping given as the elementwise soft-thresholding operator \eqref{eq:softthreshold}. Proximal mapping for the full potential is not analytically available, making ns-HMC expensive to implement.
%\begin{equation*}
   % (\prox_{\|\cdot\|_{1}^{\lambda_g}} (z))_j = \text{sgn}(z_j).\max(|z_{j}| - \lambda_g, 0)\,,
%\end{equation*}
%$j = 1, 2, \ldots, d$. 
Further details are provided in Appendix~\ref{sec: logistic_details}.

We use the \texttt{Pima.tr} dataset available in library \texttt{MASS} in \texttt{R}. It contains 200 randomly selected diabetic females in Phoenix, Arizona, from a total of 532 patients collected by US National Institute of Diabetes and Digestive and Kidney Diseases \cite{smith1988using}. It has 7 covariates and a binary response vector encoding whether diabetes was tested positive or negative. 

We implement ns-HMC with $\lambda = 1$ and p-HMC with $\lambda_g = 0.01$ (see middle panel of Figure \ref{fig:lambda-toy}) respectively, and run both chains for $N = 10^5$ iterations. We fix $L = 10$ for both ns-HMC and p-HMC, with $\varepsilon = 0.00014$ for ns-HMC and $\varepsilon = 0.00192$ for p-HMC. Figure~\ref{fig:slog_acg} compares the autocorrelation functions of both chains, demonstrating the improved performance of the p-HMC sampler. 
\begin{figure}
    \centering
    \includegraphics[width=.8\linewidth]{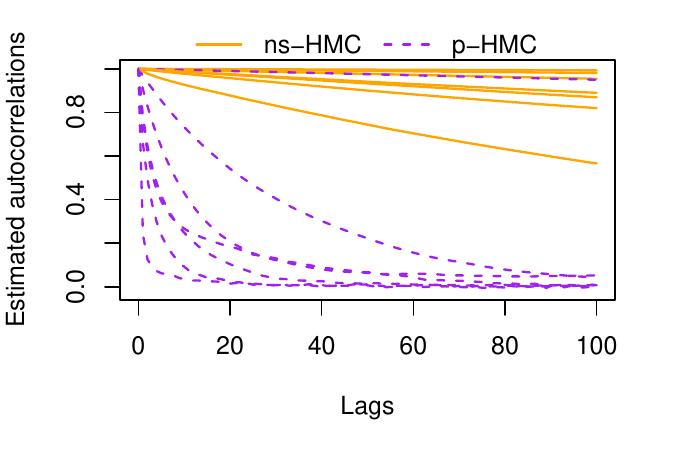}
    \caption{Estimated autocorrelation functions for ns-HMC and p-HMC for all marginal components. }
    \label{fig:slog_acg}
\end{figure}
For additional comparisons, we run a random walk Metropolis with a Gaussian proposal (RWM), the p-MALA of \cite{pereyra2016proximal}, and the my-MALA of \cite{durmus2022proximal}.  Each Markov chain was run for $N = 10^5$ iterations starting from the MAP. To further compare the performance, we repeat the simulations 100 times and report the  median, minimum, and maximum effective sample size (ESS) per second averaged over the replications of all the components. Table \ref{table; slr_ess} reports these ESS comparisons. Note that ns-HMC has the lowest ESS per second among all  algorithms across all variables and p-HMC outperforms it significantly.
% It was implemented in R and both the algorithms ns-HMC and p-HMC were compared. The Markov chain was run for a total of 5e6 iterations for both. We computed the acf plots and the ESS evaluations for the purposes of comparing the efficiency of the Markov chains with respect to each other. The results show that the p-HMC is far more efficient compared to ns-HMC both in terms of computational cost and the mixing of the Markov chains. The univariate ESS is shown on the left below whereas a comparison of the same on the log scale is shown on the right. This depicts the superiority of the p-HMC method over that of the ns-HMC. Moreover, ns-HMC takes over about 4 hours (14993 seconds exactly) whereas p-HMC method is completed within 29 minutes (1708 seconds)! This is a huge improvement in computational time. In addition, both algorithms have a multivariate ESS of 3984 and 313089 for ns-HMC and p-HMC respectively. This makes their respective ESS per unit time to be 15.943 and 10998.314 per minute respectively. 

\begin{minipage}[h]{0.5\textwidth}
%\vspace{0pt} % <-- Force top alignment
\setlength{\tabcolsep}{6pt}
\renewcommand{\arraystretch}{1.4}
\small
\begin{center}
\begin{tabular}{|l|r|r|r|}
  \hline
   Method    &   min  &   median  &    max \\ \hline
RWM   &  4.819 & 95.063 &  503.129 \\
p-HMC  &  3.239 & 454.372 &  3556.195 \\
my-MALA  & 2.199  & 22.598 &  128.345 \\
ns-HMC  & 0.002  &  0.013  &   0.065 \\
p-MALA  & 0.113  &  0.922   &  4.661 \\ 
  \hline
\end{tabular}
\captionof{table}{Minimum, median, and maximum ESS per second for all components, averaged over 100 replications.}
\label{table; slr_ess}
\end{center}
\end{minipage}%
\hfill
% \begin{minipage}[t]{0.42\textwidth}
% \vspace{0pt} % <-- Force top alignment
% \centering
% \includegraphics[width=\textwidth]{SLR_barplot.pdf}
% \captionof{figure}{Comparison of ns-HMC and p-HMC across components.}
% \end{minipage}

% \begin{figure} [H]
%     \centering
%     \includegraphics[scale = 0.5]{density_plot-1.pdf}
%     \caption{Marginal density plots with the vertical lines denoting the mode of the posterior distribution for the respective component.}
%     \label{fig:density}
% \end{figure}

% \begin{figure} [H]
%     \centering
%     \includegraphics[width = 3.25in]{acf_plot-1.pdf}
%     \caption{Component wise autocorrelation plots for both the algorithms.}
%     \label{fig:acf}
% \end{figure}

\subsection{Nuclear-norm based low-rank matrix estimation}

Low-rank matrix estimation is an important problem in matrix recovery problems, e.g., matrix completion and  image denoising. The nuclear norm prior has enjoyed widespread use in these problems  \cite{CandesPlan2010, gu2014weighted,shamsi2015image, nejati2016denoising}. 
 
 We observe a matrix $Y \in \Real^{m \times k}$ where $Y = X + E$. We wish to recover the latent low-rank matrix $X$ that has been corrupted by a noise matrix $E$ with iid Gaussian entries, i.e., $e_{ij} \overset{\text{iid}}{\sim} N(0, \sigma^2)$, where $\sigma^2 > 0$ is known. The assumed prior on $X$ is,
\begin{eqnarray} 
\label{eq:nucl_norm_prior}
    \pi_0(X)  & \propto & \exp(-\alpha\|X\|_{*})\,,
\end{eqnarray}
where $\|X\|_{*}$ is the nuclear norm of $X$ and $\alpha > 0$ is constant. Recall that $\|X\|_{*}$ is the sum of the singular values of $X$. 
The posterior density of $X$ given $Y$ is 
\begin{eqnarray} 
\label{eq:matrix_post}
    \pi(X|Y) & \propto & \exp\left\{-\left(\frac{\|Y - X\|_{\text{F}}^{2}}{2\sigma^{2}} + \alpha\|X\|_{*} \right) \right\}\,,
\end{eqnarray}
where $\|X\|_{\text{F}}$ denotes the Frobenius norm of $X$. We split the potential in \eqref{eq:matrix_post} as the sum of $f(X) :=  \|Y - X\|_{\text{F}}^{2}/2\sigma^{2}$ and $g(X) := \alpha\|X\|_*$. For p-HMC, the proximal mapping of $g$ is given by,
\begin{equation*}
    \prox_{g}^{\lambda_g} (X) = \text{SVT}(X, \alpha\lambda_g)
\end{equation*}
where $\text{SVT}(X, \tau)$ denotes the singular value soft-thresholding operator on matrix $X$ with threshold $\tau$. Let $X = B\Sigma_{\tau} D\Tra$, where $\Sigma_{\tau}$ is the diagonal matrix of singular values and let $s_i$ denote the $i^{\text{th}}$ corresponding singular value. Then $\text{SVT}(X, \tau) = B\Sigma_{\tau} D\Tra$ where $\Sigma_{\tau}$ denotes the matrix obtained by replacing each singular value $s_i$ by $\max(s_i - \tau, 0)$. Similarly for ns-HMC, the proximal mapping for $U = f + g$ is \cite{recht2010guaranteed},
\begin{equation*}
    \prox_{U}^{\lambda}(X) = \text{SVT} \left( \dfrac{\lambda}{\lambda +  \sigma^{2}}Y + \dfrac{\sigma^2}{\lambda + \sigma^2}X, \dfrac{\alpha\lambda\sigma^{2}}{\lambda + \sigma^{2}}  \right)\,.
\end{equation*}

We use a $64 \times 64$ pixel checkerboard image similar to the one in \cite{pereyra2016proximal}
for our analysis. The dimension of the posterior distribution is $64^2 (= 4096)$ with $\sigma^2 = 0.01$ and $\alpha = 1.15/\sigma^2$ as chosen in \cite{pereyra2016proximal}. Following the guidance given in Section \ref{sec:optimal_lambda}, we choose $\lambda_g = 0.0001$. As in prior examples, we run  ns-HMC with $\lambda = 1$ and  run the chains for a total of $N = 10^5$ iterations. Table~\ref{table: nn_ess} compares the ESS of ns-HMC and p-HMC with that of RWM, p-MALA and my-MALA. We see that p-HMC has the highest ESS per second for all the minimum, median and the maximum variable as described below. Evidence for p-HMC's better  mixing is  corroborated by the autocorrelation plots in Figure \ref{fig:nn_ess} which shows the gains in mixing of p-HMC over ns-HMC for all the components. 

Figure \ref{fig:checker} describes variants of the checkerboard image considered for this problem. The first three images display the latent matrix $X$, observed matrix $Y$ and the denoised MAP estimate respectively. The last image visualizes the uncertainty in the estimated pixels. For each pixel a credible interval is constructed. The pixels having the widest intervals (95th percentile) are plotted in white, allowing an assessment of which pixels demonstrate larger posterior uncertainty.
\begin{minipage}[t]{0.45\textwidth}
\vspace{0pt} % <-- Force top alignment
\setlength{\tabcolsep}{6pt}
\renewcommand{\arraystretch}{1.4}
\small
\begin{center}
\begin{tabular}{|l|r|r|r|}
  \hline
   Method    &   min  &   median  &    max \\ \hline
RWM    &    0.169   &  0.217   &  1.406 \\
pHMC    &  10.545   & 11.500   & 12.798 \\
myMALA   &  0.696   &  0.791   &  0.907 \\
nsHMC   &   0.004   &  0.005   &  0.007 \\ 
pMALA    &  0.640   &  0.726   &  0.829 \\
  \hline
\end{tabular}
\captionof{table}{Minimum, median, and maximum ESS per second for all components, averaged over 100 replications.}
\label{table: nn_ess}
\end{center}
\end{minipage}%

\begin{figure}
    \centering
    \includegraphics[width=.8\linewidth]{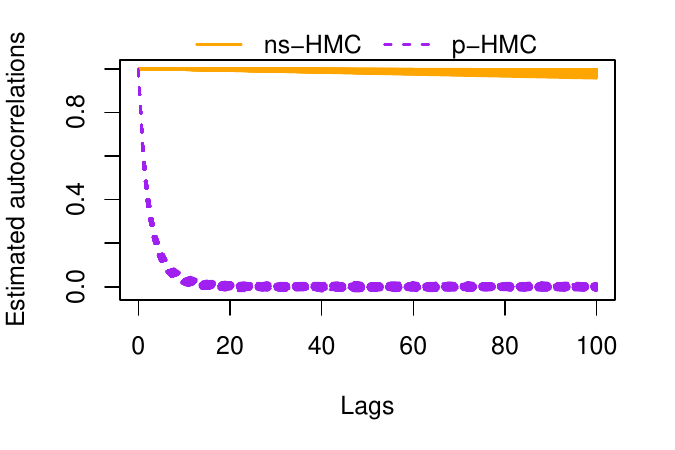}      
    \caption{Nuclear Norm: Estimated autocorrelation functions for ns-HMC and p-HMC for all marginal components.}
    \label{fig:nn_ess}
\end{figure}

\begin{figure*}
    \centering
    \includegraphics[width=.8\linewidth]{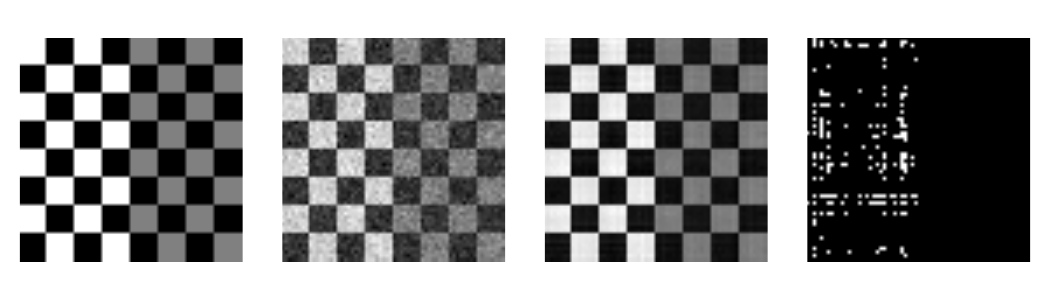}
    \caption{(From left) Original matrix, noisy matrix, MAP estimator, top 5$\%$ elements with the largest credible intervals, in white.}
    \label{fig:checker}
\end{figure*}

\section{Conclusion} \label{sec:conclusion}

This work explores an efficient proximal HMC method for non-smooth targets. It is useful for non-differentiable target densities, which are ubiquitous in signal/image processing. Leveraging the properties of proximal mappings, the described methodology achieves significant computational gains. In cases where the computational cost between the proposed p-HMC method and ns-HMC is the same, our method achieves significantly more accurate results as measured by the effective sample size. This is demonstrated in the examples where our method successfully outperforms in either the computational cost or the estimation accuracy. We also present novel theoretical analysis for p-HMC. Specifically, we provide verifiable conditions for geometric ergodicity of the p-HMC chain. Finally, we provide practical guidance on choosing the smoothing parameter $\lambda$ that tunes the degree of approximation to the target potential.

% \section{Acknowledgements}
% The authors would like to thank the IITK-Rice collaborative grant or project grant under which most of this work was undertaken. Dootika Vats would like to thank \textcolor{red}{name of grant if any}. Further, Eric Chi is also upported by \textcolor{red}{name of grant if any}. 

\appendices
\section{~}
\subsection{Proof of Theorem \ref{thm:geom_erg}} \label{ge_thm_proof}
We first state useful results from \cite{MR4003576}.

\begin{theorem}[\cite{MR4003576}]
\label{thm:living_hmc}
    An HMC algorithm produces a $\pi$-geometrically ergodic Markov chain if i) Assumption~\ref{ass:L} holds and ii) for $1/2 < \delta < 1$,
    \begin{eqnarray} \label{limsup_condn}
        \limsup_{\|x\| \rightarrow \infty, \|z\| \leq \|x\|^{\delta}} \left\{\|m_{L, \varepsilon}(x, z)\| - \|x\|\right\}  <  -\sqrt{2}L\varepsilon \eta(d)\,,
    \end{eqnarray}
    where $\eta(d) = \Gamma((d+1)/2)/\Gamma(d/2)$,
    and iii)
    \begin{eqnarray} 
    \label{eq:inwards_condn_hmc}
         \lim_{\|x\| \rightarrow \infty} \int_{R(x) \cap I(x)} q_{\text{H}}(x, y) dy & = & 0 \,.
    \end{eqnarray}
\end{theorem}

Condition \eqref{limsup_condn} is challenging to verify in practice, so \cite{MR4003576} provide the following sufficient condition that we use to prove Theorem~\ref{thm:geom_erg}.

\begin{theorem} [\cite{MR4003576}] 
\label{thm:suff_condn_hmc}
    For any $L \geq 1$, \eqref{limsup_condn} holds if
    \begin{align}
        &\text(a) \qquad \lim_{\|x\| \rightarrow \infty} \|\nabla U(x)\| \amp = \amp \infty\,, \label{eq:living1} \\
        &\text(b) \qquad \liminf_{\|x\| \rightarrow \infty} \frac{\langle \nabla U(x), x \rangle}{\|\nabla U(x)\| \|x\|} \amp > \amp 0\,,\label{eq:living2} \\
        &\text(c) 
        \qquad \limsup_{\|x\| \rightarrow \infty} \frac{\|\nabla U(x)\|}{\|x\|} \amp = \amp S_{l}\,, \label{eq:living3also}
        \end{align}
    for some $S_{l} < \infty$, then there exists an $\varepsilon_{0} < \infty$ such that \eqref{limsup_condn} holds provided $\varepsilon < \varepsilon_{0}$.
\end{theorem}
\begin{proof}
[Proof of Theorem~\ref{thm:geom_erg}]
It suffices to prove that \eqref{eq:living1}, \eqref{eq:living2} and \eqref{eq:living3also} hold under Assumption \ref{ass:living_limsup}. We first verify \eqref{eq:living1}. By parts (a) and (b) of Assumption \ref{ass:living_limsup} 
\begin{align*}
    \lim_{\|x\| \to \infty} \|\nabla U^{\lambda_g}(x)\| &= \lim_{\|x\| \to \infty} \|\nabla f(x) + \nabla g^{\lambda_g}(x)\| \\
    &\geq \lim_{\|x\| \to \infty} \|\nabla f(x)\| - \lim_{\|x\| \to \infty} \|g^{\lambda_g}(x)\| \\
    &= \lim_{\|x\| \to \infty} \|\nabla f(x)\| \left(1 - \frac{\|\nabla g^{\lambda_g}(x)\|}{\|\nabla f(x)\|}\right) \\
    &\rightarrow \quad \infty.
\end{align*}
We next verify \eqref{eq:living2}. Note that
\begin{align} \label{eq: liminf_ineq_con2}
    &\liminf_{\|x\| \to \infty} \frac{\langle \nabla U^{\lambda_g}(x), x\rangle}{\|\nabla U^{\lambda_g}(x)\| \|x\|} \nonumber \\
    &\quad \geq \liminf_{\|x\| \to \infty} \frac{\langle \nabla f(x), x\rangle}{\|\nabla U^{\lambda_g}(x)\| \|x\|} +  \liminf_{\|x\| \to \infty} \frac{\langle \nabla g^{\lambda_g}(x), x\rangle}{\|\nabla U^{\lambda_g}(x)\| \|x\|}\,.
\end{align}
We derive bounds on the two terms on the right hand side of the inequality in \eqref{eq: liminf_ineq_con2}. We first show that
\begin{equation} \label{eq: condn_g_ineq2}
    \liminf_{\|x\| \to \infty} \frac{\langle \nabla g^{\lambda_g}(x), x\rangle}{\|\nabla U^{\lambda_g}(x)\| \|x\|} \amp \geq \amp 0\,.
\end{equation}
Since $g$ is convex, its MY-envelope $g^{\lambda_g}$ is also convex. Recall that the gradient of a differentiable convex function is monotone, i.e., for any $x$ and $y$ in $\Real^d$,
\begin{equation}
\label{eq:monotone}
    \left\langle \nabla g^{\lambda_g}(x) - \nabla g^{\lambda_g}(y), x - y \right\rangle \geq 0.
\end{equation}
Taking $y = 0$ in \eqref{eq:monotone} gives
\begin{equation*}
        \left\langle \nabla g^{\lambda_g}(x), x \right\rangle \geq   \left  \langle \nabla g^{\lambda_g}(0), x \right\rangle\,.
\end{equation*}
Therefore,
\begin{align}
 \label{eq:liminf_ineq_rhs_term2}
    \liminf_{\|x\| \to \infty} \frac{\left\langle \nabla g^{\lambda_g}(x), x\right\rangle}{\|\nabla U^{\lambda_g}(x)\| \|x\|}  &\geq  \liminf_{\|x\| \to \infty} \frac{\left\langle \nabla g^{\lambda_g}(0), x \right\rangle}{\|\nabla U^{\lambda_g}(x)\| \|x\|}\\
     &\geq  - \liminf_{\|x\| \to \infty} \frac{\|\nabla g^{\lambda_g}(0)\|\|x\|}{\|\nabla U^{\lambda_g}(x)\| \|x\|} \\
     &= 0\,.
\end{align}
Now we turn our attention to bounding the second term in \eqref{eq: liminf_ineq_con2}. Note that
\begin{align*}
    &\liminf_{\|x\| \to \infty} \frac{\langle \nabla f(x), x\rangle}{\|\nabla U^{\lambda_g}(x)\| \|x\|}  =  \liminf_{\|x\| \to \infty} \frac{\langle \nabla f(x), x\rangle}{\|\nabla f(x) + \nabla g^{\lambda_g}(x)\|\|x\|} \\
    &\geq  \liminf_{\|x\| \to \infty}\left(\frac{1}{(1 + \|\nabla g^{\lambda_g}(x)\|/\|\nabla f(x)\|)}\cdot\frac{\langle \nabla f(x), x\rangle}{\|\nabla f(x)\|\|x\|}\right)\,.
\end{align*}
By parts (b) and (c) of Assumption \ref{ass:living_limsup}, the right hand side of the inequality in \eqref{eq:liminf_ineq_rhs_term2} is strictly positive which together with \eqref{eq: condn_g_ineq2} establishes \eqref{eq:living2}.

Finally, we verify \eqref{eq:living3also}. By parts (b) and (d) of Assumption~\ref{ass:living_limsup},
\begin{align*}
    &\frac{\|\nabla U^{\lambda_g}(x)\|}{\|x\|} \leq  \frac{\|\nabla f(x)\|}{\|x\|} + \frac{\|\nabla g^{\lambda_g}(x)\|}{\|x\|} \\
    &\Rightarrow \limsup_{\|x\| \to \infty}\, \frac{\|\nabla U^{\lambda_g}(x)\|}{\|x\|} \\
    &\leq \limsup_{\|x\| \to \infty}\, \frac{\|\nabla f(x)\|}{\|x\|} + \limsup_{\|x\| \to \infty}\, \frac{\|\nabla g^{\lambda_g}(x)\|}{\|\nabla f(x)\|}\frac{\|\nabla f(x)\|}{\|x\|} \\
    &\leq \limsup_{\|x\| \to \infty}\, \frac{\|\nabla f(x)\|}{\|x\|} \\
    & \quad + \lim_{\|x\| \to \infty}\frac{\|\nabla g^{\lambda_g}(x)\|}{\|\nabla f(x)\|}\cdot \limsup_{\|x\| \to \infty}\frac{\|\nabla f(x)\|}{\|x\|} \\
    &= \amp \left(1 + \lim_{\|x\| \to \infty}\frac{\|\nabla g^{\lambda_g}(x)\|}{\|\nabla f(x)\|}\right) \limsup_{\|x\| \to \infty}\frac{\|\nabla f(x)\|}{\|x\|}\\
    &< \quad \infty\,.
\end{align*}
Consequently, by Theorem~\ref{thm:suff_condn_hmc} the p-HMC algorithm is geometrically ergodic under Assumptions~\ref{ass:L} and \ref{ass:living_limsup}.
\end{proof}

\subsection{Details of Section \ref{sec:logistic}}
\label{sec: logistic_details}
We derive the expression of the posterior density for the sparse logistic regression example and also provide details on employing ns-HMC. Let $\pi_0(\beta)$ denote the prior on $\beta$
\begin{equation}
\label{eq:logistic_lasso_prior}
    \pi_0(\beta) = \prod_{j=1}^{d}\frac{\alpha\exp{(-\alpha|\beta_{i}|)}}{2}\,, \qquad \beta \in \Real^d\,.
\end{equation}
% The posterior density can then be written as,
% \begin{equation}
% \label{eq:logistic_lasso_likelihood}
%     \pi(\beta|y) \propto \exp(-U(\beta))\,.
% \end{equation}
The resulting posterior density is
\begin{align*}
    \pi(\beta|y)  &
    % \propto  \prod_{i=1}^{n} \left(\frac{e^{x_{i}\Tra\beta}}{1 + e^{x_{i}\Tra\beta}}\right)^{y_{i}}\left(\frac{1}{1 + e^{x_{i}\Tra\beta}}\right)^{1 - y_{i}} \exp (-\alpha \|\beta\|_{1}) \\
    % &= \prod_{i=1}^{n}\left(\frac{e^{y_{i}x_{i}\Tra\beta}}{1 + e^{x_{i}\Tra\beta}}\right)\exp (-\alpha \|\beta\|_{1}) \\
    % &= \exp\left(\sum_{i=1}^{n} \log\left(\frac{e^{y_{i}x_{i}\Tra\beta}}{1 + e^{x_{i}\Tra\beta}}\right) - \alpha \|\beta\|_{1} \right) \\
    \propto  \exp \left(- \left( \sum_{i=1}^{n} \log(1 + e^{x_{i}\Tra\beta}) - y_{i}x_{i}\Tra\beta + \alpha \|\beta\|_{1}\right)\right) \\
    &=  \exp(-U(\beta))\,.
\end{align*}
Define $U(\beta) =  f(\beta) + g(\beta)\,$ where
\begin{align*}
    f(\beta) =  \sum_{i=1}^{n} \log(1 + e^{x_{i}\Tra\beta}) - y_{i}x_{i}\Tra\beta \quad \text{and} \quad g(\beta) \amp = \amp \alpha \|\beta\|_{1}\,.
\end{align*}
% We follow the setup in \cite{durmus2022proximal} for  decomposing the potential function to approximate $U$ with $U^{\lambda_g} = f + g^{\lambda_g}$.
% \begin{equation*}
%     U^{\lambda_g} \amp = \amp f + g^{\lambda_g}\,.
% \end{equation*}
% Therefore, $p^{\lambda_g}(\beta|y,x) \propto \exp\left(-U^{\lambda_g}(\beta)\right)$. Evaluating the log gradient of $p^{\lambda_g}(\beta|y,x)$ requires the gradients of $f(\beta)$ and $g^\lambda(\beta)$ which are
% \begin{equation*}
%     \nabla f(\beta) = \sum_{i=1}^{n} \left(\frac{e^{x_{i}\Tra\beta}}{1 + e^{x_{i}\Tra\beta}} - y_i\right)x_i\,,
% \end{equation*}
% and
Recall
\begin{equation*}
    \nabla g^{\lambda_g}(\beta) = \frac{\beta - \prox_{g}^{\lambda_g}(\beta)}{\lambda}\,,
\end{equation*}
where $\prox_{g}^{\lambda_g}$ is the soft-thresholding operator. We compare our approach to ns-HMC method \cite{chaari2016hamiltonian} which 
% approximates $U$ with the MY-envelope of $U$
% \begin{align}
% \label{eq:ns-HMC_MYE_U}
%     U^{\lambda}(\beta) 
%     &= \min_{z \in \Real^{d}} \bigg\{ 
%         \sum_{i=1}^{n} \left(\log(1 + e^{x_{i}^\top z}) - y_{i} x_{i}^\top z\right) \nonumber \\
%     &\qquad \qquad \qquad \qquad \qquad + \alpha \|z\|_{1} 
%         + \frac{1}{2\lambda} \|\beta - z\|^2 
%     \bigg\} \,.
% \end{align}
requires the proximal mapping of $U^\lambda$
\begin{align*}
\prox_{U}^{\lambda}(\beta) 
&= \arg \min_{z \in \Real^{d}} \bigg\{ 
    \sum_{i=1}^{n} \left(\log(1 + e^{x_{i}^\top z}) - y_{i} x_{i}^\top z\right) \\
&\qquad \qquad \qquad \qquad \qquad+ \alpha \|z\|_{1} 
    + \frac{1}{2\lambda} \|\beta - z\|^2 
\bigg\} \,.
\end{align*}
This does not have a closed form expression and requires an iterative solver. We employ the FISTA solver \cite{BeckTeboulle2009}.

 \bibliographystyle{ieeetr}
 \bibliography{ref}

\end{document}